\documentclass[10pt, draftclsnofoot, onecolumn]{IEEEtran}
\usepackage{bbm}
\usepackage{cite}
\usepackage{amsfonts}
\usepackage{mathrsfs}
\usepackage{graphicx}
\usepackage{amssymb}
\usepackage{latexsym}
\usepackage{amsmath}
\usepackage{stfloats}
\usepackage{cases}
\usepackage{setspace}
\usepackage{bigstrut}
\usepackage{bm}
\usepackage{url}
\usepackage{color}

\newtheorem{theorem}{Theorem}

\begin{document}

\title{Large Intelligent Surface-Assisted Wireless Communication Exploiting Statistical CSI}
\author{Yu Han, Wankai Tang, Shi Jin, Chao-Kai Wen, and Xiaoli Ma}
%$^\ast$National Mobile Communications Research Laboratory, Southeast University, Nanjing, China\\
%$^\dagger$Institute of Communications Engineering, National Sun Yat-sen University, Kaohsiung 804, Taiwan}
\maketitle

\begin{abstract}
Large intelligent surface (LIS)-assisted wireless communications have drawn attention worldwide. With the use of low-cost LIS on building walls, signals can be reflected by the LIS and sent out along desired directions by controlling its phases, thereby providing supplementary links for wireless communication systems. In this study, we evaluate the performance of an LIS-assisted large-scale antenna system by formulating a tight approximation of the ergodic capacity and investigate the effect of the phase shifts on the ergodic capacity in different propagation scenarios. In particular, we propose an optimal phase shift design based on the ergodic capacity approximation and statistical channel state information. Furthermore, we derive the requirement on the quantization bits of the LIS to promise an acceptable capacity degradation. Numerical results show that using the proposed phase shift design can achieve the maximum ergodic capacity, and a 2-bit quantizer is sufficient to ensure capacity degradation of no more than 1 bit/s/Hz.
\end{abstract}

\section{Introduction}\label{Sec:Introduction}

As a key feature of the fifth-generation (5G) and future mobile communications, large-scale antenna systems can produce high throughput by utilizing the spatial degrees of freedom and achieve wide cell coverage with a high-gain array \cite{Larsson2014}. With these advantages, large-scale antenna systems can satisfy the growing demands of explosive data rate in the coming 5G era. However, hindrances still occur in a large-scale antenna system due to the existence of buildings, trees, cars, and even humans. To address this problem and produce fluent user experience, a typical solution is to add new supplementary links to maintain the communication link. For example, amplify-and-forward (AF) relay can be introduced in areas with poor communication signal to receive the weak signals and then amplify and retransmit them toward the next relay hop or terminal \cite{Rankov2007}. Typically, half-duplex mode is used in an AF relay-aided wireless communication system to avoid self-interference; however, it still causes degradation of system efficiency.

In recent years, large intelligent surface (LIS) technologies have been rapidly developed \cite{Cui2014,Li2017,Zhang2018,Zhao2018,Huang2018,Wu2018}. LIS is a cheap passive artificial structure that can digitally manipulate electromagnetic waves and obtain preferable electromagnetic propagation environment with limited power consumption \cite{Cui2014,Li2017,Zhang2018}. Modern LIS includes reprogrammable metasurface, which has been suggested to replace the radio frequency module in the traditional wireless communication system and reform the wireless communication architecture, thereby providing new opportunities in future communication networks \cite{Zhao2018}. Modern LIS also includes reflector array, which can reflect the electromagnetic wave, thereby offering support to traditional wireless communication systems \cite{Huang2018}. Reflector array and AF relay work in different mechanisms to provide supplementary links. Although the incident signal at a reflector array is reflected without being scaled, the propagation environment can be improved using extremely low power consumption. Moreover, using the reflector array enables the application of the full-duplex mode without causing self-interference \cite{Wu2018}. Therefore, reflector array is a more economical and efficient choice compared with the AF relay.

In this study, we apply the reflector-array-type LIS in large-scale antenna systems to provide supplementary links between the base station (BS) and the user when the line-of-sight (LoS) path is absent in this channel. By adjusting the phases of the incident signal, LIS can reflect the signal toward the desired spatial direction. We first obtain an approximation of the ergodic capacity of the LIS-assisted large-scale antenna system when the LIS assistant link is under Rician fading condition. We find that the ergodic capacity is largely dependent on the phase shift amount introduced by the LIS. To maximize the ergodic capacity, we propose an optimal phase shift design by exploiting statistical channel state information (CSI). Furthermore, in view of the hardware impairments, we formulate the requirement on the quantization bits of the LIS to ensure an acceptable degradation of the ergodic capacity. Numerical results verify the tightness of the derived ergodic capacity approximation expression, demonstrate the effectiveness of the optimal phase shift design, and show that a 2-bit quantizer is sufficient to ensure capacity degradation of less than 1 bit/s/Hz.

\section{System Model}\label{Sec:SystemModel}

We focus on a single cell of an LIS-assisted large-scale antenna system, as shown in Fig.~\ref{Fig:scenario}. The BS is equipped with an $M$-element large uniform linear array (ULA), which serves a single-antenna user. An LIS is set up between the BS and the user, comprising $N$ reflector elements arranged in a ULA.

\begin{figure}
  \centering
  \includegraphics[scale=0.8]{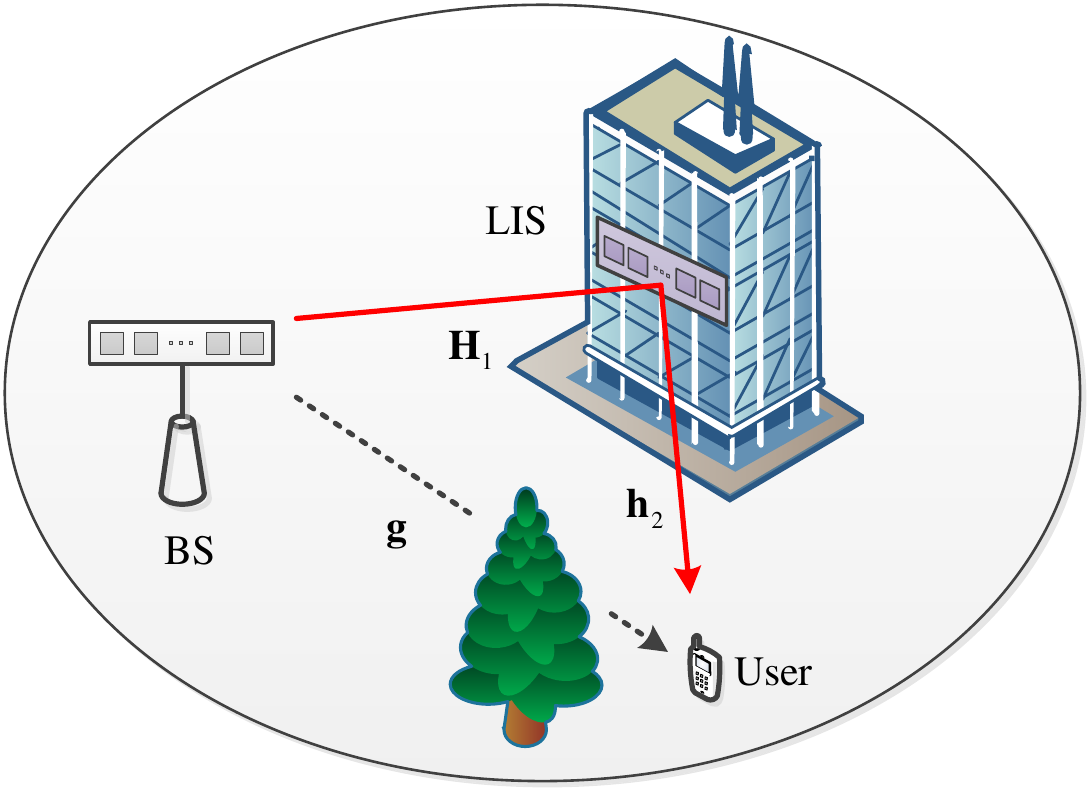}
  \caption{LIS-assisted large-scale antenna system. The LoS component between the BS and the user may be blocked. Meanwhile, LoS components exist between the BS and the LIS and between the LIS and the user.} \label{Fig:scenario}
\end{figure}

The LoS path between the BS and the user may be blocked; however, the wireless channel is filled with plenty of scatters. We model the propagation environment between the BS and the user as Rayleigh fading and denote the channel as ${\bf g}\in \mathbb{C}^{1\times M}$, where the elements of $\bf g$ are i.i.d. in complex Gaussian distribution with zero mean and unit variance. For the channel between the BS and the LIS and that between the LIS and the user, LoS components exist in practical implementation. Therefore, these two channels are modeled in Rician fading. We denote the channel between the BS and the LIS as
\begin{equation}\label{Eq:channelH1}
{\bf H}_1 = \sqrt{\frac{K_1}{K_1+1}}\bar{\bf H}_1 + \sqrt{\frac{1}{K_1+1}}\tilde{\bf H}_1,
\end{equation}
where $K_1$ is the Rician $K$-factor of ${\bf H}_1$; $\bar{\bf H}_1 \in \mathbb{C}^{N\times M}$ is the LoS component, which remains unchanged within the channel coherence time; and $\tilde{\bf H}_1 \in \mathbb{C}^{N\times M}$ is the non-LoS (NLoS) component. The elements of $\tilde{\bf H}_1$ are i.i.d. complex Gaussian distributed, and each element has zero mean and unit variance. Similarly, the channel between the LIS and the user is expressed as
\begin{equation}\label{Eq:channelh2}
{\bf h}_2 = \sqrt{\frac{K_2}{K_2+1}}\bar{\bf h}_2 + \sqrt{\frac{1}{K_2+1}}\tilde{\bf h}_2,
\end{equation}
where $K_2$ is the Rician $K$-factor of ${\bf h}_2$, $\bar{\bf h}_2 \in \mathbb{C}^{1\times N}$ is the LoS component, and $\tilde{\bf h}_2 \in \mathbb{C}^{1\times N}$ is the NLoS component. Each element of $\tilde{\bf h}_2$ is i.i.d. complex Gaussian distributed with zero mean and unit variance.

The LoS components are expressed by the responses of the ULA. The array response of an $N$-element ULA is
\begin{equation}\label{Eq:ULAsteer}
{\bf a}_N(\theta) = \left[1,e^{j2\pi \frac{d}{\lambda}\sin\theta},\ldots,e^{j2\pi \frac{d}{\lambda}(N-1)\sin\theta}\right],
\end{equation}
where $\theta$ is the angle of departure (AoD) or angle of arrival (AoA) of a signal. Under this condition, the LoS component $\bar{\bf H}_1$ is expressed as
\begin{equation}\label{Eq:H1LoS}
\bar{\bf H}_1 = {\bf a}^H_N(\theta_{{\rm AoA},1}){\bf a}_M(\theta_{{\rm AoD},1}),
\end{equation}
where $\theta_{{\rm AoD},1}$ is the AoD from the ULA at the BS, and $\theta_{{\rm AoA},1}$ is the AoA to the ULA at the LIS. Similarly, the LoS component $\bar{\bf h}_2$ is
\begin{equation}\label{Eq:h2LoS}
\bar{\bf h}_2 = {\bf a}_N(\theta_{{\rm AoD},2}),
\end{equation}
where $\theta_{{\rm AoD},2}$ is the AoD from the ULA at the LIS.

We focus on the downlink of the LIS-assisted large-scale antenna system. The signal can travel along $\bf g$ directly, or be reflected by the LIS. When reflected by the LIS, the phases of the signals are changed. The received signal at the user side is expressed as
\begin{equation}\label{Eq:signalmodel}
r = \sqrt{P}({\bf h}_2 {\bf\Phi} {\bf H}_1+{\bf g}){\bf f}^Hs + w,
\end{equation}
where $P$ is the transmit power, ${\bf\Phi} = {\rm diag}\{e^{j\phi_1},\ldots,e^{j\phi_N}\}$, $\phi_n \in [0,2\pi)$ is the phase shift introduced by the $n$th element of the LIS, ${\bf f}\in\mathbb{C}^{1\times M}$ is the beamforming vector satisfying $\|{\bf f}\|^2=1$, $s$ is original signal satisfying $\mathbb{E}\{|s|^2\}=1$, $\mathbb{E}\{\cdot\}$ represents taking expectation, and $w$ is the complex Gaussian noise with zero mean and unit variance. We denote ${\bf h}={\bf h}_2 {\bf\Phi} {\bf H}_1$ as the assistant channel. We also view ${\bf\Phi}$ as the beamforming weight introduced by the LIS.
Maximum ratio transmitting is adopted to enhance the signal power. Then, assuming ${\bf h}+{\bf g}$ is known at BS, the beamforming vector ${\bf f}$ is defined as
\begin{equation}\label{Eq:MRT}
{\bf f} = \frac{{\bf h}+{\bf g}}{\|{\bf h}+{\bf g}\|}.
\end{equation}
Note that $({\bf h}+{\bf g})$ can be estimated by the BS using the pilots sent from the user. By applying \eqref{Eq:MRT}, we obtain the ergodic capacity of the LIS-assisted large-scale antenna system as
\begin{equation}\label{Eq:ergcapacity}
C = \mathbb{E}\left\{\log_2 \left(1+ \frac{P}{\sigma_w^2}\|{\bf h}_2 {\bf\Phi} {\bf H}_1+{\bf g}\|^2\right)\right\}.
\end{equation}
where $\sigma_w^2$ is the noise variance.

We aim to identify the optimal phase shift design ${\bf\Phi}$ at the LIS to maximize the ergodic capacity. This long-term design utilizes the statistical CSI and can remain unchanged within the channel coherence time.

\section{Ergodic Capacity Analysis}\label{Sec:CapacityAnalysis}

Before designing the phase shift, we first theoretically analyze the ergodic capacity of the LIS-assisted large-scale antenna system and then evaluate the effects of using different phase shift amounts under different propagation scenarios.

\subsection{Approximation of Ergodic Capacity}

To have a direct cognition on the ergodic capacity, we provide an approximation expression in the following theorem.

\begin{theorem}\label{theorem1}
The ergodic capacity of the LIS-assisted large-scale antenna system can be approximated as
\begin{equation}\label{Eq:Theorem}
C \approx \log_2 \left(1+P\left(\gamma_1 \|\bar{\bf h}_2 {\bf\Phi} \bar{\bf H}_1\|^2+\gamma_2 MN+M\right)\right),
\end{equation}
where
\begin{equation}\label{Eq:gamma1}
\gamma_1 = \frac{K_1K_2}{(K_1+1)(K_2+1)},
\end{equation}
and
\begin{equation}\label{Eq:gamma2}
\gamma_2 = \frac{K_1+K_2+1}{(K_1+1)(K_2+1)}.
\end{equation}
\end{theorem}

\begin{proof}
According to the characteristic of the logarithmic function, it holds that
\begin{equation}\label{Eq:logcharacter}
\mathbb{E}\left\{\log_2\left(1+x\right)\right\} \approx \log_2\left(1+\mathbb{E}\left\{x\right\}\right).
\end{equation}
Besides, $\sigma_w^2=1$. Hence, \eqref{Eq:ergcapacity} satisfies
\begin{equation}\label{Eq:Theoproof1}
C \approx \log_2 \left(1+P \mathbb{E}\left\{\|{\bf h}_2 {\bf\Phi} {\bf H}_1+{\bf g}\|^2\right\}\right).
\end{equation}
Then, we focus on the derivation of $\mathbb{E}\{\|{\bf h}_2 {\bf\Phi} {\bf H}_1+{\bf g}\|^2\}$. We decompose $\|{\bf h}_2 {\bf\Phi} {\bf H}_1+{\bf g}\|^2$ by  
\begin{equation}\label{Eq:Theoproof2}
\left\|{\bf h}_2 {\bf\Phi} {\bf H}_1+{\bf g}\right\|^2 = \left\|\sqrt{\frac{1}{(K_1+1)(K_2+1)}} \left(\underbrace{\sqrt{K_1K_2}\bar{\bf h}_2 {\bf\Phi} \bar{\bf H}_1}_{{\bf x}_1} + \underbrace{\sqrt{K_1}\tilde{\bf h}_2 {\bf\Phi} \bar{\bf H}_1}_{{\bf x}_2} + \underbrace{\sqrt{K_2}\bar{\bf h}_2 {\bf\Phi} \tilde{\bf H}_1}_{{\bf x}_3} + \underbrace{\tilde{\bf h}_2 {\bf\Phi} \tilde{\bf H}_1}_{{\bf x}_4}\right)+{\bf g}\right\|^2.
\end{equation}
The first item ${\bf x}_1$ is constant, and $\mathbb{E} \{{\bf x}_i\} = {\bf 0}$ holds for $i=2,3,4$.
Since $\tilde{\bf H}_1$, $\tilde{\bf h}_2$ and ${\bf g}$ have zero means and are independent with each other, we can derive that
\begin{equation}\label{Eq:Theoproof7}
\mathbb{E}\left\{\|{\bf h}_2 {\bf\Phi} {\bf H}_1+{\bf g}\|^2\right\} = \mathbb{E} \left\{ \|{\bf g}\|^2\right\}+ \frac{ \|{\bf x}_1\|^2 + \mathbb{E} \left\{ \|{\bf x}_2\|^2+\|{\bf x}_3\|^2+\|{\bf x}_4\|^2\right\} }{(K_1+1)(K_2+1)} .
\end{equation}
For the channel between the BS and the user, it holds that
\begin{equation}\label{Eq:Theoproof8}
\mathbb{E} \left\{ \|{\bf g}\|^2\right\} = M.
\end{equation}
For the assistant channel, we derive from \eqref{Eq:H1LoS} that
\begin{equation}\label{Eq:Theoproof9}
\mathbb{E} \left\{ \|{\bf x}_2\|^2\right\} = K_1\mathbb{E} \left\{ \tilde{\bf h}_2 {\bf\Phi} \bar{\bf H}_1\bar{\bf H}_1^H{\bf\Phi}^H \tilde{\bf h}_2^H\right\}=K_1 \|\bar{\bf H}_1\|^2_F = K_1 MN.
\end{equation}
Similarly,
\begin{equation}\label{Eq:Theoproof10}
\mathbb{E} \left\{ \|{\bf x}_3\|^2\right\} = K_2 \bar{\bf h}_2 {\bf\Phi} \mathbb{E} \left\{\tilde{\bf H}_1 \tilde{\bf H}_1^H\right\} {\bf\Phi}^H \bar{\bf h}_2^H
= K_2 M \|\bar{\bf h}_2\|^2 = K_2 MN,
\end{equation}
and
\begin{equation}\label{Eq:Theoproof11}
\mathbb{E} \left\{ \|{\bf x}_4\|^2\right\} = \mathbb{E} \left\{\tilde{\bf h}_2 {\bf\Phi} \mathbb{E} \left\{\tilde{\bf H}_1 \tilde{\bf H}_1^H\right\} {\bf\Phi}^H \tilde{\bf h}_2^H\right\}
=M\mathbb{E} \left\{\|\tilde{\bf h}_2\|^2\right\} = MN.
\end{equation}
By applying \eqref{Eq:Theoproof8}-\eqref{Eq:Theoproof11} into \eqref{Eq:Theoproof7}, we obtain
\begin{equation}\label{Eq:Theoproof12}
\mathbb{E}\left\{\|{\bf h}_2 \Phi {\bf H}_1+{\bf g}\|^2\right\} = \gamma_1 \|\bar{\bf h}_2 {\bf\Phi} \bar{\bf H}_1\|^2+\gamma_2 MN+M.
\end{equation}
Finally, \eqref{Eq:Theorem} is obtained.
\end{proof}

\emph{Theorem \ref{theorem1}} indicates that when the transmit power and the LoS components remain unchanged, the ergodic capacity of the LIS-assisted large-scale antenna system is determined by $\gamma_1$, $\gamma_2$, and $\|\bar{\bf h}_2 {\bf\Phi} \bar{\bf H}_1\|^2$, which are further determined by the Rician-$K$ factors in the assistant channel and the phase shift amounts at the LIS.

\subsection{Effects of Rician-$K$ Factors and Phase Shifts}

To acquire deep insights on the effects of the Rician-$K$ factors and the phase shift amounts on the ergodic capacity, we investigate the ergodic capacity under the following special cases.

\emph{Case 1}: If $K_1 = 0$ or $K_2 = 0$, then the ergodic capacity of the LIS-assisted large-scale antenna system approximates
\begin{equation}\label{Eq:Case1}
C \approx \log_2 \left(1+PM(N+1)\right).
\end{equation}

We observe that when either ${\bf H}_1$ or ${\bf h}_2$ is under Rayleigh fading condition, the ergodic capacity is proportional to the number of elements in the large-scale antenna array and that in the LIS but is independent of the phase shift amounts at the LIS. This phenomenon is caused by the spatial isotropy that holds upon the assistant channel, which is insensitive to the beamforming between ${\bf H}_1$ and ${\bf h}_2$. Under this condition, even 0-bit phase-shifting is sufficient, and the phase shift amount can be set arbitrarily.

\emph{Case 2}: If $K_1, K_2 \to \infty$, then the ergodic capacity of the LIS-assisted large-scale antenna system approaches
\begin{equation}\label{Eq:Case2}
C \to \log_2 \left(1+P(\| \bar{\bf h}_2 {\bf\Phi} \bar{\bf H}_1 \|^2+M)\right).
\end{equation}

In the extreme Rician fading condition, only LoS components exist, and the assistant channel remains unchanged. The ergodic capacity increases in proportion to $\| \bar{\bf h}_2 {\bf\Phi} \bar{\bf H}_1 \|^2$. That is to say, in spatially directional propagation environment, the ergodic capacity of the LIS-assisted large-scale antenna is sensitive to the beamforming weights at the LIS.

If the amount of phase shifts is properly set, then the wireless signal can be beamformed on the main lobe of the assistant channel. Then, $\| \bar{\bf h}_2 {\bf\Phi} \bar{\bf H}_1 \|^2 \gg N$, and the ergodic capacity under extreme Rician fading condition is considerably higher than the ergodic capacity under Rayleigh fading condition when comparing \eqref{Eq:Case1} with \eqref{Eq:Case2}. Otherwise, when $\| \bar{\bf h}_2 {\bf\Phi} \bar{\bf H}_1 \|^2 < N$, the ergodic capacity under extreme Rician fading condition is even inferior to the ergodic capacity under Rayleigh fading condition.
Therefore, in Rician fading condition, the amount of phase shifts should be carefully designed to fully utilize the LoS components of the assistant channel $\bf h$.

\section{Phase Shift Design of Reflector Array}\label{Sec:PhaseShift}

In this section, we propose an optimal design of ${\bf\Phi}$ to maximize the ergodic capacity exploiting the statistical CSI and provide a criterion of the quantization bit to ensure an acceptable ergodic capacity degradation of the LIS-assisted large-scale antenna system.

\subsection{Optimal Phase Shift Design}

On the basis of the theoretical results in \emph{Section \ref{Sec:CapacityAnalysis}}, maximizing the ergodic capacity can be translated to maximize $\| \bar{\bf h}_2 {\bf\Phi} \bar{\bf H}_1 \|^2$. Thus, the optimal ${\bf\Phi}$ satisfies
\begin{equation}\label{Eq:optPhi}
{\bf\Phi}_{\rm opt} = \max_{{\bf\Phi}} \| \bar{\bf h}_2 {\bf\Phi} \bar{\bf H}_1 \|^2.
\end{equation}
Applying \eqref{Eq:H1LoS} and \eqref{Eq:h2LoS} into \eqref{Eq:optPhi} yields
\begin{equation}\label{Eq:optPhi1}
{\bf\Phi}_{\rm opt} = \max_{{\bf\Phi}} \| {\bf a}_N(\theta_{{\rm AoD},2}) {\bf\Phi} {\bf a}^H_N(\theta_{{\rm AoA},1}){\bf a}_M(\theta_{{\rm AoD},1}) \|^2.
\end{equation}
Denote $z = {\bf a}_N(\theta_{{\rm AoD},2}) {\bf\Phi} {\bf a}^H_N(\theta_{{\rm AoA},1})$. Then, \eqref{Eq:optPhi1} can be rewritten as
\begin{equation}\label{Eq:optPhi2}
{\bf\Phi}_{\rm opt} = \max_{{\bf\Phi}} \| z{\bf a}_M(\theta_{{\rm AoD},1}) \|^2 = \max_{{\bf\Phi}} |z|^2 \| {\bf a}_M(\theta_{{\rm AoD},1}) \|^2,
\end{equation}
where $\| {\bf a}_M(\theta_{{\rm AoD},1}) \|^2 = M$ is a constant. Hence, the optimal ${\bf\Phi}$ can also maximize $|z|^2$.
We derive that
\begin{equation}\label{Eq:z}
z = \sum_{n=1}^N {e^{j2\pi \frac{d}{\lambda}(n-1)(\sin\theta_{{\rm AoD},2}-\sin\theta_{{\rm AoA},1})+j\phi_n }}.
\end{equation}

The ergodic capacity depends on $\theta_{{\rm AoA},1}$ and $\theta_{{\rm AoD},2}$ but is independent of $\theta_{{\rm AoD},1}$. Hence, the phase shift can only affect the links that are directly connected by the LIS. Meanwhile, $0\le|z|\le N$. If $|z| = 0$, then the LoS component of the assistant channel is completely blocked due to the inadequate setting of the phase shift amount. Then, the power transmitted along the LoS component is wasted, and increasing the number of reflectors only makes limited contributions to the ergodic capacity which is gained by the NLoS component. This phenomenon further demonstrates the significance of a proper phase shift design.

Therefore, in order to maximize $|z|^2$, the optimal phase shift on the $n$th reflector element of the LIS should be
\begin{equation}\label{Eq:optphin}
\phi_{{\rm opt},n} = 2\pi \frac{d}{\lambda}(n-1)(\sin\theta_{{\rm AoA},1}-\sin\theta_{{\rm AoD},2}).
\end{equation}
That is, the wireless signal is reflected and sent out along the LoS component of the assistant channel\footnote{Note that the acquisition of $\theta_{{\rm AoA},1}$ and $\theta_{{\rm AoD},2}$ at the LIS is left for future work. One possible solution is to measure $\theta_{{\rm AoA},1}$ when initially set up the LIS, and then $\theta_{{\rm AoD},2}$ using pilots sent from the user.}.
When adopting the optimal phase shift design, $z=N$ holds, and
\begin{equation}\label{Eq:Rmaxapp}
C_{\rm max} \approx  \log_2 \left(1+PM(\gamma_1 N^2+\gamma_2 N+1)\right).
\end{equation}
The comparison of \eqref{Eq:Rmaxapp} and \eqref{Eq:Case1} indicates that the existence of LoS component is beneficial at all times if the phase shift amount is well-designed. Moreover, the ergodic capacity is proportional to $N$. Increasing the number of reflectors helps enhancing the receiving power at the user side, thereby further improving the ergodic capacity of the LIS-assisted large-scale antenna system.

\subsection{Influence of Bit Quantization}

In practical systems, the phase shift amount is constrained by the quantization bits of the LIS. We denote the number of quantization bits as $B$. Then, each theoretical value $\phi_n$ is quantized to its nearest value in
\begin{equation}\label{Eq:QuanValue}
\left\{0, \frac{2\pi}{2^B},\ldots,\frac{2\pi(2^B-1)}{2^B}\right\}.
\end{equation}
The following theorem provides guidance for the selection of LISs with different phase shift precisions.

\begin{theorem}\label{theorem2}
To promise an acceptable ergodic capacity degradation of $\xi$ bits/s/Hz compared with using optimal full-resolution phase shift amounts, the number of quantization bits of the LIS should satisfy
\begin{equation}\label{Eq:Brequire}
B\ge \log_2\pi -\log_2\arccos\sqrt{(1+\alpha N^{-2})2^{-\xi}-\alpha N^{-2}},
\end{equation}
where
\begin{equation}\label{Eq:alpha}
\alpha = \frac{1}{\gamma_1}\left(\gamma_2 N+1+\frac{1}{PM}\right).
\end{equation}
\end{theorem}

\begin{proof}
Based on \eqref{Eq:optPhi2} and \eqref{Eq:Rmaxapp}, we write the approximated ergodic capacity under bit quantization constraint as
\begin{equation}\label{Eq:Rbqapp}
C_{\rm BQ} \approx  \log_2 \left(1+PM(\gamma_1 |z|^2+ \gamma_2 N+1)\right).
\end{equation}
For the $n$th element of the LIS, we denote the quantization error of the phase shift amount as $\delta_n$, which satisfies
\begin{equation}\label{Eq:QuanErrorValue}
-\frac{2\pi}{2^{B+1}} \le \delta_n \le \frac{2\pi}{2^{B+1}}, n = 1,\ldots,N.
\end{equation}
In this condition, when the optimal phase shift amounts are employed, it holds that $z=\sum_{n=1}^N {e^{j\delta_n }}$.
Suppose that $N$ is even. Since $B \ge 1$, according to \eqref{Eq:QuanErrorValue}, we obtain
\begin{equation}\label{Eq:xQuanError1}
|z|^2 \ge \left| \frac{N}{2}( e^{j\frac{2\pi}{2^{B+1}}}+e^{-j\frac{2\pi}{2^{B+1}}})\right|^2 = N^2 \cos^2\left( \frac{\pi}{2^{B}} \right).
\end{equation}
It is obvious that degradation occurs under bit quantization constraint because $C_{\rm BQ} \le C_{\rm max}$.
We denote $C_{\rm max}-C_{\rm BQ}$ as the degradation amount of the ergodic capacity. Then, we discuss how many bits are sufficient to promise the acceptable degradation of $\xi$, i.e., $C_{\rm max}-C_{\rm BQ} \le \xi$.
According to \eqref{Eq:Rmaxapp} and \eqref{Eq:Rbqapp}, we derive that
\begin{equation}\label{Eq:xQuanError2}
|z|^2 \ge(N^2+\alpha)2^{-\xi}-\alpha,
\end{equation}
Recalling \eqref{Eq:xQuanError1}, we obtain the requirement on the quantization bit as
\begin{equation}\label{Eq:bitrequire}
N^2 \cos^2\left( \frac{\pi}{2^{B}} \right) \ge(N^2+\alpha)2^{-\xi}-\alpha,
\end{equation}
which is further translated to \eqref{Eq:Brequire}.
\end{proof}

\emph{Theorem \ref{theorem2}} indicates that if the assistant channel is under Rayleigh fading condition, that is, $\alpha\to \infty$, then \eqref{Eq:bitrequire} always holds whatever the value of $B$ is. This conclusion is in accordance with the analytical results in Case 1 of \emph{Section \ref{Sec:CapacityAnalysis}}. Moreover, when $\alpha>0$ remains unchanged, if $\xi =0$, then $B\ge\infty$. This crucial requirement is released gradually as $\xi$ increases. These reasonable findings demonstrate the correctness of \emph{Theorem \ref{theorem2}}.

Then, according to \emph{Theorem \ref{theorem2}}, the minimum value of $B$ is inversely proportional to either $P$, $M$, or $N$. Thus, the LIS-assisted large-scale antenna system becomes less sensitive to the bit quantization constraint if the transmit power is enhanced or either the antenna array or the LIS is enlarged. We take $M=N=64$, $K_1=K_2=10$, and $P=0$ dB as an example. When $\xi=1$, $B\ge2$ holds. Consequently, the ergodic capacity degradation does not exceed 1 bit/s/Hz as well, if we remain $B=2$ and set $P=10$ dB, $N=128$. That is, the 2-bit phase shift is sufficient to ensure ergodic capacity degradation of less than 1 bit/s/Hz.

\section{Numerical Results}\label{Sec:Results}

In this section, we examine the tightness of the approximation of the ergodic capacity and evaluate the optimal phase shift design for the LIS-assisted large-scale antenna system. We set $M=64$, and $K_1=K_2$. The transmit signal-to-noise ratio (SNR) equals 10 dB. Angles of the LoS components in the assistant channel is randomly set within $[0,2\pi)$.

We initially test the tightness of the ergodic capacity approximation presented in \emph{Theorem \ref{theorem1}}. Here, the optimal phase shift is adopted. We take 10,000 times of Monte Carlo simulations. The results are shown in Fig.~\ref{Fig:Simulation1}. The approximations are tightly closed with the Monte Carlo results. With the increase of the Rician $K$-factor, the gap between the Monte Carlo results and the approximations diminishes. When the Rician $K$-factor grows to infinity, the ergodic capacity approaches a constant, as discussed in \emph{Case 2}. These results demonstrate the correctness of \emph{Theorem \ref{theorem1}}. Hence, designing the phase shift amount on the basis of the approximations is reliable.

\begin{figure}
  \centering
  \includegraphics[scale=0.6]{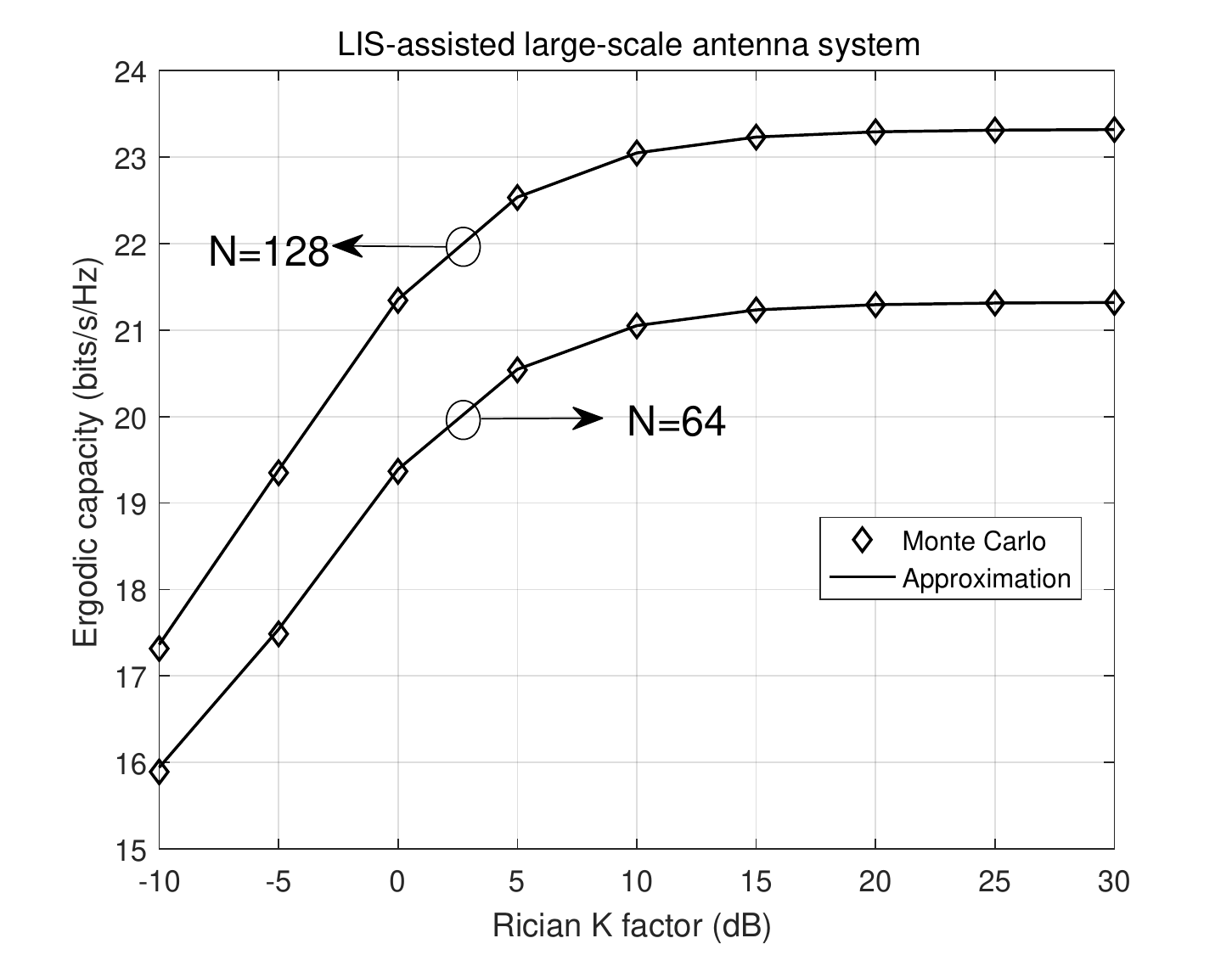}
  \caption{Comparison of the Monte Carlo results and the approximations of the ergodic capacity.} \label{Fig:Simulation1}
\end{figure}

Next, we test the effectiveness of the proposed optimal phase shift design. For comparison, we evaluate the case when the phase shift amount is randomly set, and the results of 10,000 types of random amount are collected. We further investigate the ergodic capacity under Rayleigh fading condition of $K_1=K_2=10$ dB. The results are shown in Fig.~\ref{Fig:Simulation2}. Firstly, the ergodic capacity of the optimal phase shift design is considerably higher than that of using random phase shift amount. When the number of reflectors increases, this advantage is further enlarged. Moreover, using random phase shifts under Rician fading condition achieves even lower ergodic capacity than under Rayleigh fading condition because using inadequate phase shift amount may damage the assistant channel. Therefore, proper design of phase shift amount is essential in the LIS-assisted large-scale antenna system.

\begin{figure}
  \centering
  \includegraphics[scale=0.6]{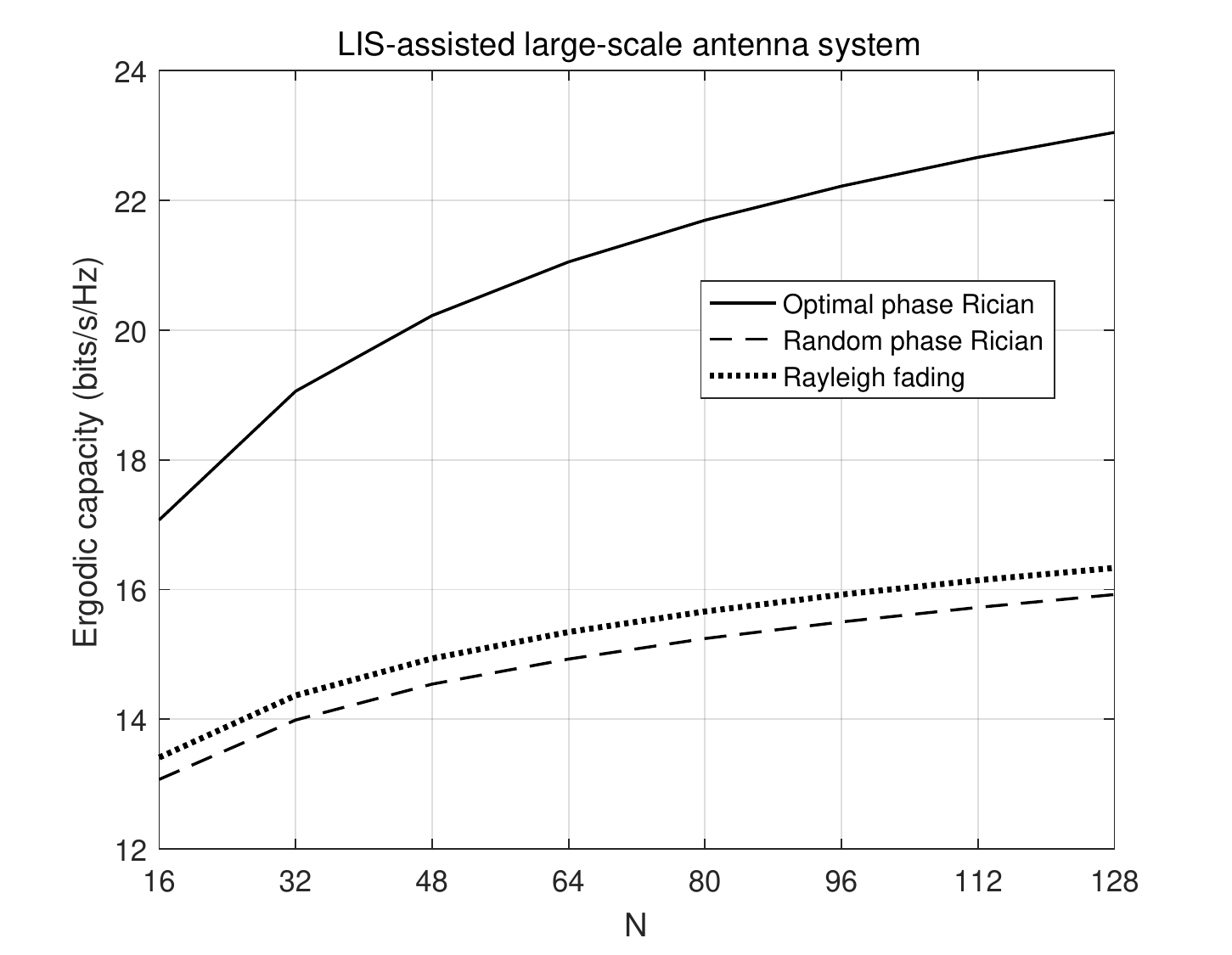}
  \caption{Comparison of ergodic capacity performances when adopting the optimal and random phase shift amounts under Rician fading condition and the ergodic capacity performance under Rayleigh fading condition.} \label{Fig:Simulation2}
\end{figure}

Finally, we test the performance degradation when bit quantization is considered upon the phase shifts in practical systems. As shown in Fig.~\ref{Fig:Simulation3}, the ergodic capacity decreases more than 1 bit/s/Hz if 1-bit quantization is adopted. The performance gap between the ergodic capacity of using perfect phase shifts and that of using quantized phase shifts decreases with the increase of the quantization bit. The performance degradation is below 1 bit/s/Hz when using 2-bit quantization, which is in accordance with the theoretical results derived from \emph{Theorem \ref{theorem2}}.

\begin{figure}
  \centering
  \includegraphics[scale=0.6]{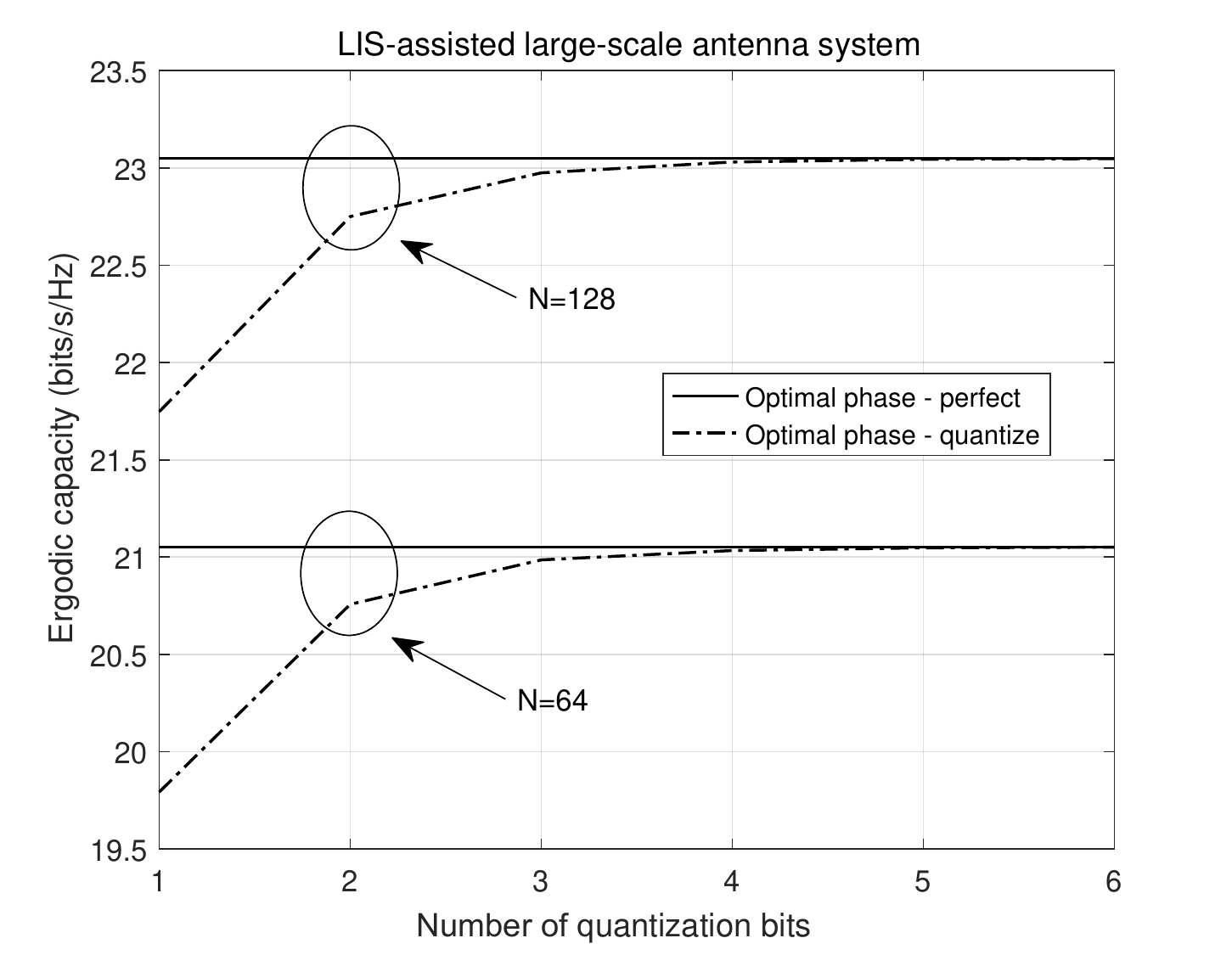}
  \caption{Comparison of ergodic capacity performances under different bit quantization constraints.} \label{Fig:Simulation3}
\end{figure}

\section{Conclusion}

In this study, we evaluated the ergodic capacity of the LIS-assisted large-scale antenna system by formulating an approximation expression and discussed the significance of using a proper phase shift design based on this expression. Particularly, we proposed an optimal phase shift design to maximize the ergodic capacity and obtained the requirement on the quantization bits to ensure an acceptable capacity degradation.
Tightness of the approximation was verified through Monte Carlo simulations, and the proposed phase shift design was proved to considerably enhance the capacity performance than using random phase shifts. Moreover, the numerical results showed that 2-bit quantization can sufficiently guarantee high capacity, which is in accordance with the analytical results.


\begin{thebibliography}{99}
\bibitem{Larsson2014}
E. Larsson, O. Edfors, F. Tufvesson, and T. L. Marzetta, ``Massive MIMO for next generation wireless systems,'' \emph{IEEE Commun. Mag.}, vol. 52, no. 2, pp. 186-195, Feb. 2014.

\bibitem{Rankov2007}
B. Rankov, and A. Wittneben, ``Spectral efficient protocols for half-duplex fading relay channels,'' \emph{IEEE J. Sel. Areas Commun.}, vol. 25, no. 2, pp. 379-389, Feb. 2007.

\bibitem{Cui2014}
T. J. Cui, M. Q. Qi, X. Wan, J. Zhao, and Q. Cheng, ``Coding metamaterials, digital metamaterials and programmable metamaterials,'' \emph{Light Sci. Appl.}, vol. 3, pp. 1-9, Oct. 2014.

\bibitem{Li2017}
L. Li, T. J. Cui, W. Ji, S. Liu, J. Ding, X. Wan, Y. B. Li, M. Jiang, C. W. Qiu, and S. Zhang, ``Electromagnetic reprogrammable coding-metasurface holograms,'' \emph{Nat. Commun.}, vol. 8, pp. 1-7, Aug. 2017.

\bibitem{Zhang2018}
L. Zhang, X. Q. Chen, S. Liu, Q. Zhang, J. Zhao, J. Y. Dai, G. D. Bai, X. Wan, Q. Cheng, G. Castaldi, V. Galdi, and T. J. Cui, ``Space-time-coding digital metasurfaces,'' \emph{Nat. Commun.}, vol. 9, pp. 1-11, Oct. 2018.

\bibitem{Zhao2018}
J. Zhao, X. Yang, J. Y. Dai, Q. Cheng, X. Li, N. H. Qi, J. C. Ke, G. D. Bai, S. Liu, S. Jin, and T. J. Cui, ``Controlling spectral energies of all harmonics in programmable way using time-domain digital coding metasurface,'' arXiv:1806.04414, Jun. 2018.

\bibitem{Huang2018}
C. Huang, G. C. Alexandropoulos, A. Zappone, M. Debbah, and C. Yuen, ``Energy efficient multi-user MISO communication using low resolution large intelligent surfaces,'' arXiv:1809.05397, Sep. 2018.

\bibitem{Wu2018}
Q. Wu, and R. Zhang, ``Intelligent reflecting surface enhanced wireless network: Joint active and passive beamforming design,'' arXiv:1809.01423, Sep. 2018.

\end{thebibliography}
\end{document}